\documentclass{conm-p-l}

\usepackage[active]{srcltx}

\usepackage{euscript}
  {}
  {}

\newtheorem{theorem}{Theorem}
\newtheorem{lemma}[theorem]{Lemma}

\newtheorem{prop}[theorem]{Proposition}

\newtheorem{maintheorem}{Main Theorem}

\newcommand{\dif}{\mathrm{d}}

\theoremstyle{definition}

\theoremstyle{remark}

\newcommand{\norm}[1]{\left\Vert#1\right\Vert}
\newcommand{\R}{\mathbb R}

\newcommand{\nablash}{\nabla{\kern -.75 em
     \raise 1.5 true pt\hbox{{\bf/}}}\kern +.1 em}
\newcommand{\Deltash}{\Delta{\kern -.69 em
     \raise .2 true pt\hbox{{\bf/}}}\kern +.1 em}
\newcommand{\Rslash}{R{\kern -.60 em
     \raise 1.5 true pt\hbox{{\bf/}}}\kern +.1 em}

\newcommand{\gammab}{\bar\gamma}

\newcommand{\chih}{\hat{\chi}}

\newcommand{\Sphere}{\mathbb S}
\newcommand{\D}{\partial}

\title[Black hole initial data]{Black hole initial data with a horizon of prescribed intrinsic and extrinsic geometry}
\author{Brian Smith}

\address{
  Freie Universit\"at Berlin, Arnimallee 3,
  14195 Berlin, Germany}

\email{bsmith@math.fu-berlin.de}



\thanks{This research was funded by the Deutsche 
Forschungsgemeinschaft project SFB 647 B-4 (Space - Time - Matter. Analytic and Geometric Structures)}

\begin{document}
\begin{abstract}
The purpose of this work is to construct asymptotically flat,
time-symmetric initial data with an apparent horizon of   prescribed intrinsic and extrinsic
geometry.  To do this, we use the parabolic partial differential
equation for prescribing scalar curvature. In this equation the
horizon geometry is contained within the freely specifiable part
of the metric.  This contrasts with the conformal method in which
the geometry of the horizon can only be specified up to a
conformal factor.
\end{abstract}

\maketitle

Asymptotically flat, time-symmetric initial data in general
relativity consists of an asymptotically flat Riemannian manifold
$(M,g)$ with non-negative scalar curvature $R=R(g)\geq 0$.  The
interpretation of the scalar curvature is that $R=16\pi\rho$,
where $\rho$ is the local mass density.  Within the initial data
set, a minimal surface that is outermost with respect
to one of the ends of the manifold represents a black hole since
it bounds the largest possible set that, according to cosmic
censorship, would have to be contained within the black hole
region of the Cauchy development of that end. Such a surface is
generally referred to as an \textit{apparent horizon}, and this
can only consist of a union of topological spheres~\cite{H}.  In
the present work, we shall use this term to indicate one such
component.

In the time-symmetric case, \textit{black hole initial data} shall
be here taken  to consist of one end of such a manifold external
to the region bounded by the apparent horizons.  Henceforth, we use the
notation $(M,g)$ to indicate this, and so, with the assumption that
there is only one horizon, $(M,g)$ is now taken to
mean a manifold with apparent horizon boundary $S=\D M\approx
\Sphere^2$. This is sufficient for the Cauchy problem since the
region interior to the apparent horizons does not affect the
Cauchy development exterior to the black hole region of this end.  
For more on black hole initial data see the references in~\cite{S}.

Relevant geometric information concerning an apparent horizon
consists of its metric $h_0$, trace free second fundamental form
$\chih_0$, and the local mass density restricted to the horizon
$\rho_0=\rho|_{S}$.  In the present work we shall say that such a
horizon is \textit{described} by $(\Sphere^2,h_0,\chih_0,\rho_0)$.
Since an apparent horizon must be an outermost minimal surface, it
must be  area minimizing in that any normal variation of
the surface must increase area. By the second variation of area
formula\footnote{This is actually from the standard second variation formula 
combined with the Gauss equation.  
For more on the second variation formula see~\cite{ch.} pp. 169-171.}, 
this implies that for any non-zero
 $C^{\infty}$ test function $\varphi$, one
must have
\[
    \int_{\Sphere^2}\left( |\nabla\varphi|_{h_0}^2+\left(\kappa(h_0)-\frac{1}{2}|\chih_0|^2-8\pi\rho_0\right)\varphi^2\right)\ \dif A_{h_0}\geq 0,
\]
where $\kappa(h_0)$ is the Gauss curvature of $h_0$. We shall be
concerned with the slightly stronger situation in which
\[
        \int_{\Sphere^2} \left(|\nabla\varphi|_{h_0}^2+\left(\kappa(h_0)-\frac{1}{2}|\chih_0|^2-8\pi\rho_0\right)\varphi^2\right)\, \dif A_{h_0}>0
\]
for any non-zero $C^{\infty}$ test function $\varphi$,
in which case we shall say that the horizon is \textit{strictly
stable}.  Thus, in a sense strict stability is almost necessary.
In this work we prove that this condition is  sufficient for the
existence of asymptotically flat, time-symmetric initial data
 containing an apparent horizon described by
$(\Sphere^2,h_0,\chih_0,\rho_0)$.  That is, there holds
\begin{maintheorem}
Let $h_0$ be a Riemannian metric on $\Sphere^2$.  Let $\chih_0$ be
a symmetric tensorfield of rank $(0,2)$, which is trace free with
respect to $h_0$.  Let $\rho_0$ be a nonnegative function.  Then
there exists asymptotically flat, time-symmetric initial data with
a strictly stable horizon boundary described by
$(\Sphere^2,h_0,\chih_0,\rho_0)$ provided  the  operator
\[
    \Delta_{h_0}-\left(\kappa(h_0)-\frac{1}{2}|\chih_0|^2-8\pi\rho_0\right)
\]
is negative.
\end{maintheorem}

As in ~\cite{S}, the main theorem is proved by taking
$M=[r_0,\infty)\times\Sphere^2$, along with an appropriate
local mass density $\rho$, and family of metrics $h(r)$ on
$\Sphere^2$, and solving for a function $u$ such that
\[
        g=u^2dr^2+h
\]
satisfies $R(g)=16\pi\rho$.  This is accomplished using the
parabolic scalar curvature equation\footnote{For more on this equation and its use see
~\cite{bartnik91},\cite{bartnik93},\cite{bartnik03}, \cite{Sh},\cite{ST},
\cite{SW}, \cite{SW2},\cite{WY}. }, which we use here in the
following form:
\[
       \tilde H\D_ru=u^2\Delta_hu-\left(\kappa(h)-8\pi\rho\right)u^3+\left(\D_r\tilde H+\frac{\tilde
       H^2+|\tilde\chi|^2}{2}\right)u,
\]
where $\tilde\chi_{AB}=\frac{1}{2}\frac{\D h_{AB}}{\D r}$ and
$\tilde H=h^{AB}\tilde\chi_{AB}$.  Concerning the foliation spheres
$S_r=\{r\}\times\Sphere^2$, the extrinsic curvature $\chi$ and mean
curvature $H$ are related to $\tilde\chi$ and $\tilde H$ by
$\tilde\chi=u\chi$ and $\tilde H=uH$.

The strength of using the parabolic scalar curvature equation 
 for our purposes is that it allows
for the control of the geometry  of the foliation $S_r$.
Indeed, since the mean curvature of this
foliation is $H=\tilde H/u$,  by simply choosing the
family $h$ such that $\tilde H>0$ for $r>r_0$, a minimal surface at $S_{r_0}$
must be outermost, and hence an apparent horizon. In addition, we
can control the intrinsic geometry of $S_{r_0}$ simply by the
choice of $h(r_0)=h_0$.  It should be noted that these features of
the general method were also of central importance in~\cite{S}.

The difference between the present work and the previous work lies
in how we arrange that $S_{r_0}$ is  minimal.  In~\cite{S} this
was arranged by requiring that $u\to\infty$ as $r\to r_0$.  Then
$S_{r_0}$ is, in fact, totally geodesic since the extrinsic
curvature of the foliation is given by $\chi=\tilde\chi/u$.  We cannot
use this here since we would now like $\chi$ to be non-trivial at
$r_0$, and, in fact, prescribed.  The only alternative is to
arrange that $\tilde H$ vanishes initially.  Inserting this
information into the parabolic scalar curvature equation, we see
that the initial data $u_0$ is required to satisfy the elliptic
equation
\begin{equation}\label{eq:elliptic}
        \Delta_{h_0}u_0-\left(\kappa(h_0)-\frac{1}{2}|\chih_0|^2-8\pi\rho_0\right)u_0+\frac{1}{u_0}=0,
\end{equation}
if, in addition, we assume $\tilde H=r-r_0$ on a neighborhood of
$S_{r_0}$.

Thus, the outline of the method in the present work is as follows:
We first solve Equation~\eqref{eq:elliptic} for the initial data
$u_0$.  We then choose the family $h$ to satisfy
\begin{align}
         \frac{\D h_{ij}}{\D r}h^{ij}&>0,\label{eq:h1}\\
         \frac{\D h_{ij}}{\D r}|_{r_0}&=2(\chih_0)_{ij}u_0,\\
         \frac{\D h_{ij}}{\D r}h^{ij}&=2(r-r_0),\,\, r\in
         [r_0,r_0+\varepsilon),\label{eq:h3}
\end{align}
and in addition, for $r\geq T$, $T$ large, we require
$h=r^2\gammab$, where $\gammab$ is the standard round metric on
$\Sphere^2$.  We also extend $\rho_0$ as a smooth, non-negative,
compactly supported  function $\rho$ on $M$ in such a way that we
may solve the parabolic scalar curvature equation.  We then solve
the parabolic scalar curvature equation with the initial data
$u_0$ and check that the resulting function $u$ also has
appropriate asymptotic behavior.

Solving the parabolic scalar curvature equation near the horizon
cannot, however, be achieved by a simple application of standard
parabolic theory. Indeed, the equation will not be parabolic at
$r=r_0$, and in addition, we would like for $u$ to be $C^{\infty}$
up to and including the initial surface at $r=r_0$. We are able to
deal with this in the following way: Defining $A=\left(\D_r\tilde
H+\frac{\tilde H^2+|\tilde\chi|^2}{2}\right)$ and
$L=\Delta_h-(\kappa-8\pi\rho)$, the parabolic scalar curvature
equation takes the form
\[
      t\D_tu=u^2L u+A u,
\]
where we have defined $t=(r-r_0)$.  Let $\tilde L$ and $\tilde A$
be the $n$-th order Taylor polynomials of $L$ and $A$,
respectively. We solve an $n$-th order approximation of the
parabolic scalar curvature equation
\begin{equation}\label{eq:parapprox}
    t\D_t\tilde u=\tilde u^2\tilde L\tilde u +\tilde A \tilde u-P,
\end{equation}
where $P=a_0t^{n+1}+a_1t^{n+2}+\cdots+ a_{3n-1}t^{4n}$, $a_0,a_1,\ldots,a_{3n-1}\in C^{\infty}(\Sphere^2)$.
The polynomial $P$ is chosen so that Equation~\eqref{eq:parapprox} admits a solution of the form
\[
    \tilde u=u_0+u_1t+\cdots+u_nt^n,
\]
with $u_0,u_1,\dots,u_n$ time independent.  This is possible
since, as will be seen in Section~\ref{sec:decomp}, inserting this
into~\eqref{eq:parapprox} yields a sequence
of elliptic equations for the $u_i$; the first of these is just ~\eqref{eq:elliptic}.
After having solved the sequence of elliptic equations,
$P$ is chosen to annihilate the remaining terms in Equation~\eqref{eq:parapprox}.
We then
insert the decomposition $u=\tilde u+v t^n$ into the parabolic
scalar curvature equation. We are able to solve the resulting
equation to obtain a $C^{\infty}$ solution $v$ which satisfies
$(t\D_t)^mv\in O(t)$ for $m=0,1,\ldots,n$.  The resulting solution
$u$ of the parabolic scalar curvature equation has all of the
desired properties.

The outline of the paper is as follows: In
Section~\ref{sec:elliptic} we solve Equation~\eqref{eq:elliptic}
for the function $u_0$. In Section~\ref{sec:decomp} we precisely
define the family $h$ on a small annular region $A_{\varepsilon}=B_{r_0+\varepsilon}(0)\backslash B_{r_0}(0)$
about the horizon and, as described in the previous paragraph,
we decompose the parabolic scalar curvature equation into
Equation~\eqref{eq:parapprox} and the part for the remainder $v$.
In Section~\ref{sec:sequence} we decompose \eqref{eq:parapprox}
into the sequence of elliptic equations for the $u_n$, which are
then solved. In Section~\ref{sec:remain} we solve the equation for
the remainder $v$.  In the next to the last section we establish a new
result for global existence of solutions of the parabolic scalar curvature equation.
In the final section the proof of the main theorem is completed by showing how to extend the definition of $h$
smoothly to $\R^3\backslash B_{r_0+\varepsilon}(0)$
such that the parabolic scalar curvature equation has a unique global solution such that the constructed metric
\[
  g=u^2 dr^2+h
\]
is asymptotically flat.

I would like to thank Dr. Kashif Rasul for useful discussions while preparing the manuscript and 
I would like to thank Professor Robert Bartnik for suggesting  prescribing 
the horizon extrinsic geometry as an interesting problem.

\section{The elliptic equation for the initial data}\label{sec:elliptic}
This section\footnote{This section, as well as Section~\ref{sec:sequence}, 
uses standard methods in elliptic P.D.E theory.  For more on these see~\cite{GT}.} 
deals with the  elliptic equation of the form
\begin{equation}\label{eq:vellipt}
    \Delta u-au+\frac{1}{u}=0
\end{equation}
on $(\Sphere^2,h_0)$, with $a\in C^{\infty}(\Sphere^2)$ and
$\Delta=\Delta_{h_0}$.
The main result of this section is that the
equation is solvable provided that the operator $\Delta-a$ is
strictly negative. This is a generalization of the following
theorem that was proved in~\cite{S}:
\begin{theorem}  \label{thm:vellipt}
Assume $a>0$ and let $a_*=\inf_{\Sphere^2}a,
a^*=\sup_{\Sphere^2}a$. Then Equation~\eqref{eq:vellipt} has a
positive solution $u\in C^{\infty}$ satisfying $1/\sqrt{a^*}\leq
u\leq 1/\sqrt{a_*}$. Within the class of $C^2$ functions there are
exactly two solutions $\pm u$.
\end{theorem}

We shall use the method of continuity to generalize this to the
case that $\Delta-a$ is merely assumed negative. Namely, we take a
positive function $a_0$ and define $a_{\tau}=a_0(1-\tau)+\tau a$
for $\tau\in [0,1]$ and consider the family of equations
\begin{equation} \label{eq:deform}
     \Delta u_{\tau}-a_{\tau}u_{\tau}+\frac{1}{u_{\tau}}=0.
\end{equation}
By the previous theorem, we know that a solution exists for
$\tau=0$. We must only prove that solutions continue to exist with
appropriate bounds for all $\tau\in [0,1]$.

Before formally stating and proving the resulting theorem, we note
that $a^*_{\tau}$ remains positive,  we retain the lower bound $u_{\tau}\geq 1/\sqrt{a_{\tau}^{*}}$, and
$\Delta-a_{\tau}$ remains negative.  The latter is easily enough
checked since the energy

\[
  E_{\tau}(\varphi)=\int_{\Sphere^2}\left(|\nabla\varphi|^2+a_{\tau}\varphi^2\right)\, \dif A_{h_0}
\]
satisfies
\[
     E_{\tau}(\varphi)=(1-\tau)E_0(\varphi)+\tau E_1(\varphi)\geq\delta=\min\{\delta_0,\delta_1\}>0,
\]
for all $\varphi\in H^1$  with $\norm{\varphi}_{L^2}\equiv 1$, where 
$\delta_i=\inf\{E_i(\varphi): \norm{\varphi}_{L^2}\equiv 1\},\, i=0,1$; note that each of $\delta_i$ 
has a positive lower bound by the standard argument involving Rellich compactness.  
By applying $E_{\tau}$ to $\varphi\equiv 1$
the fact that $a^*_{\tau}>0$ is easily verified.
The proof of $u_{\tau}\geq 1/\sqrt{a^{*}_{\tau}}$ is a
simple application of the maximum principle.  Indeed, any solution
$u_{\tau}\in C^2(\Sphere^2)$ of~\eqref{eq:deform} cannot change signs,
and so we may, without loss of generality, assume that $u_{\tau}>0$.  At
a point $p$ at which the infimum is attained one has $\Delta
u_{\tau}\geq 0$; whence at $p$ there holds $a_{\tau}u_{\tau}\geq 1/u_{\tau}$.  Thus,
$a_{\tau}(p)>0$ and $\inf u_{\tau}=u_{\tau}(p)\geq 1/\sqrt{a_{\tau}(p)}\geq
1/\sqrt{a_{\tau}^*}$.

We are now in a position to prove

\begin{theorem}
Suppose that the operator $\Delta-a$ is negative on $H^1$.
Then the equation
\begin{equation}
      \Delta u-a u+\frac{1}{u}=0
\end{equation}
has a  positive $C^{\infty}$ solution that is unique within the class of
positive $C^{2}$ functions.
\end{theorem}
\begin{proof}
Let $u_0$ denote the solution of Equation~\eqref{eq:deform} at
$\tau=0$.  The linearization of the associated elliptic operator
at $\tau=0$ is just $\Delta-(a_0+u_0^{-2})$, which is bijective as an
operator from $C^{2,\alpha}$ into $C^{\alpha}$, for instance.
Hence, by the implicit function theorem, classical solutions
continue to exist on some interval $[0,\tau_{(1)})$.  Assume that
$\tau_{(1)}$ is maximal in this regard, but $\tau_{(1)}<1$.

In fact, it follows from the negativity of $L_{\tau}=\Delta-a_{\tau}$ that
the solutions $u_{\tau}$ are bounded in $L^2$ on this
interval. To see this, suppose otherwise, and let $\tau_n$ be a sequence,
$\tau_n\to \tau_{(1)}$ as $n\to\infty$, such that
$\lim_{n\to\infty}\norm{u_{\tau_n}}_{L^2}=\infty$.  Defining
$\tilde u_{\tau_n}=u_{\tau_n}/\norm{u_{\tau_n}}_{L^2}$, this new
sequence of functions satisfies $\norm{\tilde u_{\tau_n}}_{L^2}\equiv 1$.
But with $\delta$ as in the paragraph before the theorem,  using~\eqref{eq:deform}
we can choose $n$ large enough that
\[
    E_{\tau_n}(\tilde u_{\tau_n})=
    -\langle L_{\tau_n}\tilde u_{\tau_n},\tilde u_{\tau_n}\rangle|_{L^2}
    =\frac{1}{\norm{u_{\tau_n}}_{L^2}^{2}}\int_{\Sphere^2} \dif A_{h_0}<\delta,
\]
which is a contradiction.

Using the equation again together with the $L^2$ bound on $u_{\tau}$, we obtain a bound
\[
     \int_{\Sphere^2}|\nabla u_{\tau}|^2\, \dif A_{h_0} = \int_{\Sphere^2} a_{\tau}u^2_{\tau}\, \dif A_{h_0} - \int_{\Sphere^2}  \dif A_{h_0} \leq C<\infty,
\]
and so  $u_{\tau}$ is uniformly bounded in $H^1$.  Thus,
we may use Rellich's compactness theorem to obtain a sequence 
$u_{\tau_n}$ that converges weakly in $H^1$ and strongly in $L^1$ to an $H^1$ function $u_1$.
Now the $u_{\tau_n}$ are essentially bounded from below by a positive constant.  This is preserved in the limit,
and we have $u^{-1}_{\tau_n}\to u_1^{-1}$ in $L^1$ as $n\to\infty$ as well.  Hence
\begin{align*}
   0&= \lim_{n\to\infty}\int_{\Sphere^2} \left(-\nabla u_{\tau_n}\cdot\nabla\varphi-a_{\tau_n}u_{\tau_n}\varphi +\frac{\varphi}{u_{\tau_n}}\right)\, \dif A_{h_0}\\
        &=\int_{\Sphere^2} \left(-\nabla u_1\cdot\nabla\varphi-a_1u_1\varphi +\frac{\varphi}{u_1}\right)\, \dif A_{h_0},
\end{align*}
for any test function $\varphi\in C^{\infty}(\Sphere^2)$,
and thus $u_1$ is seen to be an $H^1$ weak solution of~\eqref{eq:deform} at $\tau_{(1)}$, 
which is essentially bounded from below by a positive constant.  Standard elliptic regularity theory now shows this to be a $C^{\infty}$ solution.

Applying now the implicit function theorem again as above, we see
that there exists a solution on a slightly larger interval
$[0,\tau_{(1)}+\varepsilon)$, in contradiction to the assumption that
$\tau_{(1)}$ is maximal. Hence $\tau_{(1)}=1$.

To establish uniqueness within the class of positive $C^2$ functions, we consider the difference of two such solutions
$u_1,u_2$.  With $\delta u=u_1-u_2$, one has
\[
      \Delta\delta u-\left(a+\frac{1}{u_1u_2}\right)\delta u= 0.
\]
But if $\delta u\neq 0$ then $\Delta - a$ cannot be a negative operator, in contradiction to the hypothesis of the theorem.
\end{proof}

\section{Decomposition of the parabolic scalar curvature
equation near the horizon}\label{sec:decomp}

In this section we decompose the parabolic scalar curvature
equation on a region near the horizon into the part for the $n$-th
order Taylor polynomial
\[
    \tilde u=u_0+u_1t+\cdot+u_nt^n,
\]
and the part for the remainder, but first we more precisely define
the family $h(t)$.

Using the solution $u_0$ of Equation~\eqref{eq:elliptic}, whose existence is
guaranteed by the results of last section, and recalling that
$t=r-r_0$, we define $\tilde\chi_0=u_0\chih_0$ and take
$h_{ij}(t)$ as the solution of the nonlinear ordinary differential
equation
\begin{equation}\label{eq:hdiff}
    h'_{ij}=t h_{ij}+2 \left((\tilde\chi_0)_{ij}-\frac{1}{2}\text{tr}_h\tilde\chi_0h_{ij}\right).
\end{equation}
By the standard short time theory for ordinary differential
equations, this is possible on some interval $[0,\varepsilon]$,
where in addition we may assume that $h$ is $C^{\infty}$.  Note
that $h$ satisfies Conditions~\eqref{eq:h1}~-~\eqref{eq:h3} of the
introduction. We extend $\rho_0$ to
$[r_0,r_0+\varepsilon]\times\Sphere^2$ by
\[
    \rho(t,p)=\rho_0(p).
\]

We may now use the definitions of $h$ and $\rho$ to compute
$L=\Delta_h-(\kappa(h)-8\pi\rho)$ and $A=\D_r\tilde H+1/2(\tilde
H^2+|\tilde\chi|^2)$ on $[0,\varepsilon]$.  As families of
operators from $C^{k+2,\alpha}$ into $C^{k,\alpha}$, we see that
$L$ and $A$ are $C^{\infty}$ in time.  Whence we have the
expansions
\begin{align*}
      L&=L_0+L_1t+\cdots+L_nt^n+R_Lt^n\\
      A&=A_0+A_1t+\cdots+A_nt^n+R_At^n,
\end{align*}
where  $R_L,R_A$ verify
$(tD_t)^mR_L,(tD_t)^mR_A \in O(t)$ for $n=0,1,\ldots,n$. Note that
$L_0=\Delta_{h_0}-(\kappa(h_0)-8\pi\rho_0)$ and
$A_0=1+|\tilde\chi_0|^2/2=1+|\hat\chi_0|^2u_0^2/2$.

To begin to decompose the parabolic scalar curvature equation,
define $\tilde L=L_0+L_1t+\cdots+L_nt^n$ and $\tilde
A=A_0+A_1t+\cdots+A_nt^n$, and assume that $u$ is a solution of the parabolic scalar curvature equation on
$[0,\varepsilon]$ that can be expressed as
\[
   u=u_0+u_1t+\cdots+u_nt^n+vt^n;
\]
as pointed out in the introduction, $u_0$ necessarily has to be
the solution of Equation~\eqref{eq:elliptic}.  As at the beginning of the section, we take
 $\tilde
u=u_0+u_1t+\cdots+u_nt^n$.  Then
\begin{align*}
      u^2Lu&=u^2L\tilde u+t^nu^2L v=\left(\tilde u^2+2\tilde u
      vt^n+t^{2n}v^2\right)L\tilde u+t^nu^2L v\\
      &=\tilde u^2L\tilde u+t^nu^2L v+\left(2\tilde u L\tilde u
      +t^{n}v L\tilde u\right)t^nv\\
      &=\tilde u^2\tilde L\tilde u+t^nu^2L v+\left(2\tilde u L\tilde u
      +t^{n}vL\tilde u\right)t^nv+\tilde u^2t^nR_L\tilde u,
\end{align*}
and
\[
   A u = \left(\tilde A+t^nR_A\right)(\tilde u+t^n v)=\tilde
   A\tilde u + \tilde A t^n v + t^{2n}R_A v+ t^nR_A\tilde u.
\]
Thus, under the above assumptions, the parabolic scalar curvature
equation admits the decomposition
\begin{align}
      t\frac{\D\tilde u}{\D t}&=\tilde u^2\tilde L\tilde u+\tilde
      A\tilde u-P\label{eq:polynom0}\\
      t\frac{\D v}{\D t}&= u^2L v+\left(A+2\tilde u L\tilde u
      +t^nvL\tilde u-n\right)v+\left(\tilde u^2R_L\tilde u+R_A\tilde
      u+\frac{P}{t^n}\right),\label{eq:remain}
\end{align}
where we have introduced an undetermined function $P$. Our procedure shall be to determine $P$ so that we may solve~\eqref{eq:polynom0} for
$\tilde u$ an $n$-th order polynomial in time, and then prove
the existence of a solution $v$ of~\eqref{eq:remain} satisfying
$t^m\D_t^m\nabla^kv\in O(t),\, m=0,1,\ldots,n,\, k\in \mathbb N$.

\section{A degenerate non-linear parabolic equation with polynomial dependence on
time}\label{sec:sequence}

In this section we study the degenerate parabolic equation
\begin{equation}\label{eq:polynom}
     t\frac{\D u}{\D t}=u^2Lu+Au-P
\end{equation}
in the case that the source terms have polynomial expansions in
time.  More precisely, we assume that
\[
       L=L_0+L_1t+L_2t^2+\cdots+L_nt^n
\]
\[
       A=A_0+A_1t+A_2t^2+\cdots+A_nt^n,
\]
and
\[
    P=a_0t^{n+1}+a_1t^{n+2}+\cdots+a_{3n-1}t^{4n},
\]
where $L_0=\Delta_{h_0}-(\kappa(h_0)-8\pi\rho_0)$, the $L_i$
are second order differential operators on $\Sphere^2$ in general, and  $a_i,A_i\in C^{\infty}(\Sphere^2)$ for all $i$;
$A_0$ is here the same as in the last section.

We have the following theorem:
\begin{theorem}
Assume that $u_0$ is a positive $C^{\infty}$  solution of
\[
    L_0u_0+\frac{A_0}{u_0}=0,
\]
and assume furthermore, that $L_0$ is a negative
operator.
Then the $a_i$ in the polynomial $P$ may be chosen such that there exists a solution $u$ of
Equation~\eqref{eq:polynom} of the form
\[
        u=u_0+u_1t+u_2t^2+\cdots+u_nt^n,
\]
where $u_1,u_2,\dots,u_n$ are $C^{\infty}$ functions on $\Sphere^2$
and $u_0$ is as above.
\end{theorem}
\begin{proof}
We substitute $u=\Sigma_{i=0}^nt^iu_i$,  $L=\Sigma_{i=0}^nt^iL_i$,  $A=\Sigma_{i=0}^nt^iA_i$ into Equation~\eqref{eq:polynom} to obtain
\[
   \Sigma_{m=0}^{n} mu_mt^m=\Sigma_{m=0}^{4n}\Sigma_{a+b+c+d=m} u_au_bL_cu_d
   t^m+\Sigma_{m=0}^{2n}\Sigma_{a+b=m} A_au_b t^m-P.
\]
Equating coefficients for the first $n+1$ terms yields the following coupled sequence of equations:
\[
      mu_m= \Sigma_{a+b+c+d=m} u_au_bL_cu_d+\Sigma_{a+b=m}
      A_au_b.
\]
The first of these is just $u_0^2L_0u_0+A_0u_0=0$, which we have already assumed is uniquely solvable. 
For $m>0$ we separate the terms involving $u_m$ to get
\begin{gather*}
      mu_m= u_0^2L_0u_m+A_0u_m+2u_0u_mL_0u_0+\Sigma_{a+b+c+d=m;\, a,b,d<m} u_au_bL_cu_d\\+\Sigma_{a+b=m,\, b<m}A_au_b.
\end{gather*}
But since $L_0u_0=-A_0/u_0$, these can be rewritten as
\begin{gather} \label{eq:mtheq}
            L_0u_m-\left(\frac{A_0+m}{u_0^2}\right)u_m=\\ \nonumber -\frac{1}{u_0^2}\left(\Sigma_{a+b+c+d=m;\, a,b,d<m} u_au_bL_cu_d+\Sigma_{a+b=m,\, b<m}
      A_au_b\right).
\end{gather}
Note that the right hand side only contains $u_i$ for $i < m$, and so in principle these equations can be solved inductively.

To proceed with the solution, note that when $m=0$ the equation can be solved by hypothesis, and for $m>0$,
the equations are linear in $u_m$.
Furthermore, it was assumed that $L_0$ is negative; whence so is $L_0-\frac{(A_0+m)}{u_0^2}$.  Thus,
assuming that a unique $C^{k,\alpha}$ solution exists for all $i\leq m$ the right hand side of
\begin{align*}
      &L_0u_{m+1}-\left(\frac{A_0+(m+1)}{u_0^2}\right)u_{m+1}\\
      =&-\frac{1}{u_0^2}\left(\Sigma_{a+b+c+d=(m+1);\, a,b,d<(m+1)} u_au_bL_cu_d+\Sigma_{a+b=(m+1),\, b<(m+1)}
      A_au_b\right)
\end{align*}
is determined as a $C^{k-2,\alpha}$ function; whence standard elliptic regularity theory shows that there exists a unique solution $u_{m+1}\in C^{k,\alpha}$.  Thus, the sequence of $n+1$ equations is uniquely solvable for the functions $u_1,u_2,\ldots,u_n$.  Since $k$ was arbitrary, these are in fact $C^{\infty}$.

Now that we have determined the $u_1,u_2,\ldots, u_n$, the proof is completed upon defining
\[
    P=\Sigma_{m= n+1}^{4n}\Sigma_{a+b+c+d=m} u_au_bL_cu_d
   t^m+\Sigma_{m=n+1}^{2n}\Sigma_{a+b=m} A_au_b t^m.
\]
\end{proof}

\section{The parabolic equation for the remainder}\label{sec:remain}

We now need to prove that there exists a unique solution
$v$ of Equation~\eqref{eq:remain} such that $v\to 0$ as $t\to 0$.
The first step in this direction shall be an
$L^2$ bound for solutions with small initial data near $t=0$.
To do this, we separate the terms in the equation according to their
order of growth in $t$. Recall that the equation is
\[
t\frac{\D v}{\D t}= u^2L v+\left(A+2\tilde u L\tilde u
      +t^nvL\tilde u-n\right)v+\left(\tilde u^2R_L\tilde u+R_A\tilde
      u+\frac{P}{t^n}\right).
\]
using the expansions of $A,L,$ and $P$,  and recalling that $(tD_t)^mR_L,(tD_t)^mR_A\in O(t)$, $m=0,1,2,\ldots,n$,
we write this as
\[
    t\frac{\D v}{\D t}= u^2L v+\left(A_0+2u_0 L_0 u_0-n
      +\psi_1+v\psi_2\right)v+ \psi_3,
\]
where $\psi_1,\psi_2, \psi_3 $ are $C^{\infty}$ functions $\psi_i:
\Sphere^2\times\R^+\to\R$  that satisfy $(tD_t)^m\psi_i\in O(t),\, m=0,1,2,\ldots,n$ when
regarded as functions $\psi_i:\R^+\to C^k(\Sphere^2)$. But
$L_0u_0=-A/u_0$, and so we have
\begin{equation}\label{eq:remainigen}
     t\frac{\D v}{\D t}= u^2L v+\left(-A_0-n
      +\psi_1+\psi_2 v\right)v+ \psi_3.
\end{equation}
 We would
like to prove $L^2$ bounds for the solution $v$.   To that end,
the following lemma is useful:
\begin{lemma}\label{lem:basic}
Let $L$ be as above, and suppose on $[0,\varepsilon]$ there  holds
$-\int_{\Sphere^2}\varphi L\varphi \, \dif A_{h}\geq\lambda_0
\int_{\Sphere^2}\varphi^2\, \dif A_h$ for all $\varphi\in H^1(\Sphere^2)$. Let $\psi_0:\R^2\to \R$, $\psi_0\in
C^1$, be such that $\psi_0(0,0)=1$.  Suppose that $v\in H^1(\Sphere^2)$.  Then there exists
$\delta\leq\varepsilon$ such that if $|t|, |v|<\delta$ then
$-\int_{\Sphere^2}v\psi_0(t,v) Lv \, \dif A_h\geq\frac{\lambda_0}{2}
\int_{\Sphere^2}v^2 \, \dif A_h$.
\end{lemma}
The proof of this lemma is a fairly straightforward estimate and shall be omitted.
The lemma allows us to use techniques for linear  equations to
derive $L^2$ bounds, as contained in the next lemma.
\begin{lemma}
There exists $\delta,C>0$  such that if $0<t_1<\delta$ and $v$  is a
solution of Equation~\eqref{eq:remainigen} on $[t_0,t_1]$ with
$v|_{t_0}=0$ and $|v|<\delta$ then $\norm{v|_t}_{L^2(\Sphere^2)}\leq C t$.
\end{lemma}
\begin{proof}
We shall make the change of variables $t=e^{\tau}$ so that Equation~\eqref{eq:remainigen} becomes
\[
     \frac{\D v}{\D\tau}= u^2 L v -(A_0+n+O(e^{\tau})+O(e^{\tau})v)v+O(e^{\tau}).
\]
We define $\tau_0=\log t_0, \tau_1=\log t_1$.
Multiplying by $v/u^2_0$ and integrating yields
\[
    \frac{1}{2}\frac{d}{d\tau}\int_{\Sphere^2} \frac{v^2}{u_0^2}\, \dif A_h-\frac{1}{2}\int_{\Sphere^2}\frac{v^2}{u_0^2}\tilde H \, \dif A_h\leq\int_{\Sphere^2}\left( v \frac{u^2}{u_0^2}Lv+ \frac{C}{u_0^2} v e^{\tau}\right)\, \dif A_h,
\]
provided that we make an initial choice of $\delta$ such that the coefficient of the linear term is negative.
Using now the previous lemma with a possibly smaller choice of $\delta$ and using  $ve^{\tau}\leq \frac{1}{2}\left(av^2+\frac{1}{a}e^{2\tau}\right)$, we have
\begin{gather*}
     \frac{1}{2}\frac{d}{d\tau}\int_{\Sphere^2} \frac{v^2}{u_0^2}\, \dif A_h\leq -\frac{\lambda_0}{2}\int_{\Sphere^2} v^2\, \dif A_h + \frac{e^{\tau}}{2\min u_0^2}\int_{\Sphere^2} v^2\, \dif A_h\\
      +\frac{aC}{2\min u_0^2}\int_{\Sphere^2} v^2\, \dif A_h +\frac{C}{2a\min u_0^2}\max_{[\tau_0,\tau_1]}\text{area}_h(\Sphere^2) e^{2\tau},
\end{gather*}
where we have also used $\tilde H = e^{\tau}$. Thus, defining
$w=v/u_0$, choosing $a$ correctly, and choosing $\delta$  perhaps
smaller still, we have
\[
        \frac{1}{2}\frac{d}{d\tau}\int_{\Sphere^2} w^2 \, \dif A_h\leq Ce^{2\tau}.
\]
Integrating over $[\tau_0,\tau]$ we get
\[
    \frac{1}{2}\int_{\Sphere^2} w^2(\tau,\cdot) \, \dif A_h(\tau)\leq C\int_{\tau_1}^{\tau}e^{2\tau'}\, \dif \tau' \leq C\left(e^{2\tau}-e^{2\tau_0}\right)\leq Ce^{2\tau}.
\]
This yields the conclusion of the lemma since $t=e^{\tau}$.
 \end{proof}

Of course, using Moser iteration\footnote{This section, as well as the next one, makes heavy use 
of standard parabolic techniques; e.g. Moser iteration, Schauder estimates, etc. For more 
information on these see~\cite{lieberman}, for instance.  Of course, in our case these have been applied on a compact manifold.}, one can derive an analogous supremum bound:
\begin{lemma}
With the hypotheses of the last lemma, $\delta$ can be chosen such that $|v|\leq C t$ on $[t_0,t_1]$.
\end{lemma}

We may now remove the initial assumption on the smallness of $v$:
\begin{prop}
There exists $C,\varepsilon>0$ such that any solution of Equation~\eqref{eq:remain} on $[t_1,\varepsilon]$ with $v|_{t_1}\equiv 0$ satisfies
$|v|\leq C t$ on $[t_1,\varepsilon]$.
\end{prop}
\begin{proof}
Let $\delta$ and $C$ be as in the previous lemma.  Choose $\varepsilon=\delta/2 C$.  
With this choice the conclusion holds.  Indeed, suppose otherwise.  Since $v=0$ at $t_1$, 
by continuity, it is clear that the inequality must hold on some interval $[t_1,t']$; let 
this be the maximal interval for which $|v|\leq C t$. But we have at $t'$   
that $|v|_{t'}\leq \delta/2$, and so it is clear that $|v|<\delta$ for some slightly 
larger interval $[t_1,t'+\varepsilon']$; whence the hypotheses of the previous lemma are 
fulfilled, which implies that $|v|\leq Ct$ on $[t_1,t'+\varepsilon']$ as well.
\end{proof}
Applying standard regularity theory to the equation for $v$ with time variable $\tau$ 
we obtain similar bounds for the derivatives of $v$.  To state the result, let $\D_t$ 
denote the derivative in the time direction, and let $\nabla$ denote covariant 
differentiation along $\Sphere^2$.  One has
\begin{prop}
For $t_0>0$, solutions of Equation~\eqref{eq:remain} on $[t_0,T]\times\Sphere^2$ are $C^{\infty}$. Moreover,
for any integers $k,l$, there exists
$\varepsilon$ and $C$ such that any solution of Equation~\eqref{eq:remain} on $[t_1,\varepsilon]$ with $v|_{t_1}\equiv 0$ satisfies
$\norm{\nabla^k(t\D_t)^l v}_{C^0(\Sphere^2)}\leq C t$ on $[t_0,\varepsilon]$.
\end{prop}
Finally, this may be used together with the Ascoli-Arzela theorem to obtain the
existence of a solution of Equation~\eqref{eq:remain} with the desired growth properties.
\begin{theorem}
There exists $\varepsilon>0$ such that Equation~\eqref{eq:remain} has a unique
classical solution on $(0,\varepsilon]\times\Sphere^2$ in the class $v\in O(t)$.
Furthermore, for any integers $k,l$ there exists a constant $C$ such that
$\norm{\nabla^k(t\D_t)^l v}_{C^0(\Sphere^2)}\leq C t$ on $[0,\varepsilon]$.
\end{theorem}
\begin{proof}
Using the equation with $\tau$ instead as the time variable, this corresponds
to the existence of a solution on $(-\infty,\log\varepsilon]\times\Sphere^2$
that grows like $e^{\tau}$.   To obtain this, let $\varepsilon$ be as in the previous proposition,
let $\tau_n$ be a sequence with $\tau_n\to-\infty$, and let $v_n$ be a sequence of solutions on
$[\tau_n,\log\varepsilon)$ with $v_n|_{\tau_n}\equiv 0$; the latter exist from standard existence theory.
Consider now the functions $\omega_n=e^{-\tau}v_n$.  By virtue of the previous proposition, for every
integer $k$ these satisfy $\norm{\omega_n}_{C^k([\tau_n,\log\varepsilon]\times\Sphere^2)}\leq C$.  By
using the Ascoli-Arzela theorem, we may extract a subsequence $v_m$ that converges in $C^k$ on any
set $[T,\log\varepsilon]\times\Sphere^2, T>-\infty$ to a function $\omega\in C^k$.
Inserting this sequence into Equation~\eqref{eq:remain} one finds that $v=e^{\tau}\omega$ must
be a solution of Equation~\eqref{eq:remain} on $(-\infty,\log\varepsilon]\times\Sphere^2$.
Furthermore, the bounds on each $v_n$ are maintained through the limiting process so that $v$
satisfies the bounds in the conclusion.

To see that the above constructed solution is unique, let $v_1,v_2$ be two such solutions of class $O(t)$.
Then computing the equation for the difference $\delta v =v_2-v_1$, we arrive at an equation of the form
\[
   \frac{\D\delta v}{\D\tau}= \tilde u^2_2L\delta v -\left(A_0+n-\psi_1-\psi_2(v_1+v_2)-t^n(\tilde u_1+\tilde u_2)L v_1\right)\delta v,
\]
where $\tilde u_i= u=u_0+u_1t+\cdots+u_nt^n+v_it^n$, $i=1,2$.  But by virtue of the bounds on $v_1,v_2$,
 we may assume $\tau$ is small enough that
\[
    A_0+n-\psi_1-\psi_2 (v_1+v_2)-t^n(\tilde u_1+\tilde u_2)L v_1> c> 0.
\]
Furthermore, with the help of Lemma~\ref{lem:basic}, we may assume that $\tilde u_2^2 L$ is strictly negative.
Thus, multiplying by $\delta v$ and integrating, we arrive at an inequality of the form
\[
     \frac{d}{d\tau}\norm{\delta v}_{L^2(\Sphere^2)}\leq -c \norm{\delta v}_{L^2(\Sphere^2)},
\]
from which we get
\[
     \norm{\delta v|_{\tau}}_{L^2(\Sphere^2)}\leq \norm{\delta v}_{L^2(\Sphere^2)}|_{\tau_0}e^{-c(\tau-\tau_0)}.
\]
But using the assumed bounds on $v_1,v_2$ again, for any $\varepsilon'>0$, we may assume
that $\norm{\delta v|_{\tau_0}}_{L^2(\Sphere^2)}<\varepsilon'$ by assuming $\tau_0$ is small enough;
whence $\norm{\delta v|_{\tau}}_{L^2(\Sphere^2)}\leq\varepsilon'e^{-c(\tau-\tau_0)}$.  Then letting $\tau_0\to-\infty$
we obtain that $\norm{\delta v|_{\tau}}_{L^2(\Sphere^2)}=0$ for any $\tau$.
\end{proof}

\section{Global Existence for the Parabolic Scalar Curvature Equation}

In this section it is established that negativity of the operator 
$L=\Delta_h-\left(\kappa(h)-8\pi\rho\right)$ is enough to guarantee 
global existence results for the parabolic scalar curvature equation 
in the case of positive bounded initial data.
This is contained in 
Theorem~\ref{thm:exist} at the end of the section.  This result is 
proved with the help of several intermediate estimates, which are contained 
in the next two lemmata, and the proposition that follows.

The first estimate concerns the \begin{it}subcritical equation\end{it}
\begin{equation}\label{eq:subcrit}
     \tilde H \D_r u=u^2\Delta u - f u^{3- \alpha }+ Au.
\end{equation}
One is able to derive bounds for this equation in the case that $f$ is positive on
average and $f\in L^q(\Sphere^2)$ for appropriate $q$.
Specifically, there holds
\begin{lemma}\label{lem:subcrit}
Let $q>2$. Suppose that $f\in L^q(\Sphere^2)$ and
$f_0=\int_{\Sphere^2}f\, \dif A_h>0$.  Then solutions of the subcritical
equation, Equation~\eqref{eq:subcrit}, on $[r_1,r_2]$ are bounded
for any positive initial data.  That is, there holds
\[
      u\leq  C,
\]
where $C$ depends on
\begin{gather*}
\inf_{[r_1,r_2]}f_0,\,
\sup_{[r_1,r_2]}\norm{f|_r}_{L^q(\Sphere^2)},\,
\inf_{[r_1,r_2]\times\Sphere^2}\tilde H,\\
\sup_{[r_1,r_2]\times\Sphere^2}|\tilde H|,\,
\sup_{[r_1,r_2]\times\Sphere^2}|\D_r\tilde H|,\,
\sup_{[r_1,r_2]\times\Sphere^2}|A|,\,
\norm{u|_{r_1}}_{L^3(\Sphere^2)},\, \alpha,\, |r_2-r_1|,
\end{gather*}
and of course the family $h$ on $[r_1,r_2]$.
\end{lemma}
\begin{proof}
Using Moser iteration  one may establish a supremum bound as long as one has a bound on
$\norm{u}_{L^p([r_1,r_2]\times\Sphere^2)}$ for $p>2$.
Hence, it suffices for our purposes to establish an $L^3$ bound, which we now do. We shall actually
obtain a bound on $\norm{u|_r}_{L^3(\Sphere^2)}$ for each $r\in [r_1,r_2]$.

By standard elliptic theory on a compact manifold, one may write $f=f_0+\Delta_h\omega$.  
Multiplying the subcritical  equation by $u^2$ and integrating over $\Sphere^2$ we obtain
\begin{gather*}
      \int_{\Sphere^2}\frac{\tilde H}{3}\D_r u^3\, \dif A_h = \int_{\Sphere^2} \left(-4 u^3|\nabla u|^2-(f_0+\Delta_h\omega) u^{5-\alpha}+A u^3\right)\, \dif A_h\\
      =\int_{\Sphere^2} \left(-\frac{16}{25} |\nabla u^{\frac{5}{2}}|^2
      -f_0 u^{5-\alpha}+(5-\alpha)\nabla\omega\cdot u^{4-\alpha}\nabla u+A u^3\right)\, \dif A_h.\\
\end{gather*}
We now estimate the middle term in the integrand:
\begin{gather*}
(5-\alpha)\nabla\omega\cdot u^{4-\alpha}\nabla u=
(5-\alpha)u^{\frac{5}{2}-\alpha}\nabla\omega\cdot
u^{\frac{3}{2}}\nabla u\\
\leq\frac{(5-\alpha)}{2}\left(u^{5-2\alpha}\frac{1}{\varepsilon}|\nabla\omega|^2
+\varepsilon(u^{\frac{3}{2}})^2|\nabla u |^2\right)\\
=\frac{(5-\alpha)}{2}\left(u^{5-2\alpha}\frac{1}{\varepsilon}|\nabla\omega|^2
+\frac{4\varepsilon}{25}|\nabla u^{\frac{5}{2}}|^2\right).
\end{gather*}
Thus, choosing $\varepsilon=8/(5-\alpha)$ we obtain
\begin{gather*}
      \int_{\Sphere^2}\frac{\tilde H}{3}\D_r u^3\, \dif A_h\leq
       -\int_{\Sphere^2}\left(u^{5-\alpha}f_0-\frac{(5-\alpha)^2}{16}|\nabla\omega|^2 u^{5-2\alpha}\right)\, \dif A_h
       +\int_{\Sphere^2} A u^3\, \dif A_h\\
       =-\int_{\Sphere^2}u^{5-2\alpha}\left(u^{\alpha}f_0-\frac{(5-\alpha)^2}{16}|\nabla\omega|^2\right)\, \dif A_h
       +\int_{\Sphere^2} A u^3\, \dif A_h.
\end{gather*}
But, the integrand of the first term on the right can only be negative when
\[
    u < \left(\frac{|\nabla\omega|^2(5-\alpha)^2}{16 f_0}\right)^{\frac{1}{\alpha}}.
\]
Whence
\[
      \int_{\Sphere^2}\frac{\tilde H}{3}\D_r u^3\, \dif A_h\leq \int_{\Sphere^2}\left(\frac{25|\nabla\omega|^2}{16
                f_0}\right)^{\frac{5}{\alpha}-1}f_0\, \dif A_h+\int_{\Sphere^2} A u^3\, \dif A_h.
\]
Using now the fact that $|\nabla\omega(r,\cdot)|\leq C \norm{(f-f_0)|_r}_{L^q(\Sphere^2)}$ (see~\cite{AU}), we obtain
\[
         \D_r\int_{\Sphere^2} \tilde H u^3\, \dif A_h \leq C+C\int_{\Sphere^2} \tilde H u^3\, \dif A_h,
\]
where $C$ has the dependencies as in the hypothesis of the
theorem.
Integrating this yields $\int_{\Sphere^2} \tilde H u^3\dif A_h \leq C$ on $[r_1,r_2]$, 
which in turn yields the desired bound upon incorporating the lower bound for $\tilde H$ into $C$: 
$\int_{\Sphere^2} u^3\dif A_h \leq C$ on $[r_1,r_2]$.

The supremum bound is then established via Moser iteration as in the remarks at the beginning of the proof.
\end{proof}

By the next lemma, given a bound on $\norm{u}_{L^{\varepsilon}(\Sphere^2,h)}$, for some $\varepsilon>0$,
we can ``absorb'' a factor of $u^{\alpha}$ into $f$ in order to write the 
parabolic scalar curvature equation as the subcritical equation with $f$ replaced by 
$\tilde f=f u^{\alpha}$, where
it can be arranged that $\tilde f$ is positive on average.
\begin{lemma}
Let $\varepsilon>0$.
Suppose that for all $r\in [r_1,r_2]$ we have a bound $\norm{u}_{L^{\varepsilon}(\Sphere^2)}\leq C$, $u\geq c>0$.  
Let $f_1>0$ be such that $\int_{\Sphere^2} f\, \dif A_{h}\geq f_1$. Then there exists 
$0<\alpha<\varepsilon/2$, $\alpha= \alpha(c,C,f_0,f,\varepsilon)$, such  that
\[
     \int_{\Sphere^2} fu^{\alpha}\, \dif A_h\geq \frac{1}{2}\int_{\Sphere^2} f\, \dif A_h.
\]
Furthermore, for $q=\varepsilon/\alpha>2$, one has that $\norm{f u^{\alpha}}_{L^{q}(\Sphere^2)}\leq C'$, 
where $C'$ is a constant $C'=C'(c,C,f_0,f,\varepsilon)$.
\end{lemma}
\begin{proof}
For the time being, we divide the area element $\, \dif A_h$ by $\text{area}_h(\Sphere^2)$ 
so that the total measure is 1: $\dif \mu=\, \dif A_h/\text{area}_h(\Sphere^2)$.  
Up until the last line of the proof any integrals are over $\Sphere^2$ with respect to $d\mu$; this holds also for any norms.      
Thus, 
our bound on $\norm{u}_{L^{\varepsilon}(\Sphere^2)}$ becomes 
$\norm{u}_{L^{\varepsilon}}\leq C /A^{\frac{1}{\varepsilon}}$, but since $\varepsilon$ is 
here fixed, we shall merely rename our constant $C$.

Let $0<\delta<\frac{1}{2}\int f/\norm{f}_{L^1}$, and let $\alpha<
\varepsilon/2$ be a small positive number to be chosen later.  We
have
\[
\int f u^{\alpha}=\int f_+ u^{\alpha}-f_- u^{\alpha}\geq c^{\alpha} \int f_+-f_-u^{\alpha}.
\]
Let now $p$ be chosen such that $\norm{f_-}_{L^p}<\sqrt{1+\delta}\norm{f_-}_{L^1}$, and let $p'$ be
 conjugate to $p$. Then using H\"{o}lder's inequality
\[
       \int f_- u^{\alpha}\leq \norm{f_-}_{L^p}\norm{u^{\alpha}}_{L^{p'}}
       \leq \sqrt{1+\delta}\norm{f_-}_{L^1}\left(\int u^{\varepsilon\frac{\alpha p'}{\varepsilon}}\right)^{1/p'}.
\]
Assuming $\alpha$ to be chosen small enough that $\alpha p'/\varepsilon<1$, 
another application of H\"{o}lder's inequality yields
\[
    \int f_- u^{\alpha}\leq \norm{f_-}_{L^p}\norm{u^{\alpha}}_{L^{p'}}\leq \sqrt{1+\delta}\norm{f_-}_{L^1} \norm{u}_{L^{\varepsilon}}^{\alpha}.
\]
We now further assume $\alpha$ to be chosen such that
$\norm{u}_{L^{\varepsilon}}^{\alpha}<\sqrt{1+\delta}$, which we may do by
virtue of the bound on $\norm{u}_{L^{\varepsilon}}$. Making finally
the further restriction on $\alpha$ that it be small enough that
$c^{\alpha}>1-\delta$, we have
\[
    \int f u^{\alpha}>(1-\delta)\int f_+- (1+\delta)\int f_-= \int f -\delta \norm{f}_{L^1}> \int f/2.
\]

Returning now to the original measure, this inequality continues to hold.
The final part of the conclusion follows from H\"{o}lder's inequality.
\end{proof}

These two lemmata can be combined to obtain a supremum bound on
solutions  of the parabolic scalar curvature equation in terms of
a bound on $\sup_r\norm{u}_{L^{\varepsilon}(\Sphere^2,h)}$.

\begin{prop}
Let the family of metrics $h$ be defined on $[r_1,r_2]$ such that
$\Lambda^{-1}\leq\tilde H\leq \Lambda$ for a positive constant
$\Lambda$.   Then if $\norm{u}_{L^{\varepsilon}(\Sphere^2)}$ is
uniformly bounded on $[r_1,r_2]$, the supremum of $u$ is also
uniformly bounded.
\end{prop}
\begin{proof}
As remarked, we absorb a factor of $u^{\alpha}$   into $f$
for $\alpha$ as in the previous lemma. That is, we define $\tilde
f= f u^{\alpha}$, so that the parabolic scalar curvature equation
becomes
\[
   \tilde H \D_r u=u^2\Delta u - \tilde f u^{3- \alpha }+ Au,
\]
where $\int\tilde f$ is uniformly bounded from below and
$\norm{\tilde f}_{L^q(\Sphere^2)}$ is uniformly bounded with $q$ also as
in the previous lemma. Thus, applying Lemma~\ref{lem:subcrit}, the
proposition follows.
\end{proof}
We are now in a position to prove
\begin{theorem}\label{thm:exist}
Let $h$ be a family of metrics on $\Sphere^2$  satisfying $h\in C^{\infty}([r_1,r_2]\times \Sphere^2)$.  
 Let $L=\Delta_h-(\kappa(h)-8\pi\rho)$ so that the parabolic scalar curvature equation can be written
\[
     \tilde H\D_ru=u^2Lu+Au.
\]
On the interval $[r_1,r_2]$ assume that $\tilde H>h_0>0$.   Then
for any smooth positive initial data a solution exists on
$[r_1,r_2]$ provided the operator $L$ is strictly negative.
\end{theorem}

\begin{proof}
By standard parabolic theory, it is enough to establish  upper and
lower pointwise bounds on $u$.  Since the parabolic scalar
curvature equation always admits a positive lower bound on any
finite interval, we must only prove a supremum bound, which we now
do.

For convenience, put $B=\kappa(h)-8\pi\rho$.  For $0<\sigma<1$,
multiply the parabolic scalar curvature equation by $u^{-\sigma}$
and integrate to obtain
\[
    \frac{1}{1-\sigma}\int_{\Sphere^2}\! \tilde H\D_{r}u^{1-\sigma}\, \dif A_h
    = -\int_{\Sphere^2} \! \left((2-\sigma) u^{1-\sigma}|\nabla u|^2+Bu^{3-\sigma}\right)\, \dif A_h+\int_{\Sphere^2}
    \! Au^{1-\sigma}\, \dif A_h.
\]
But this is equivalent to
\begin{gather*}
      \frac{1}{1-\sigma}\frac{d}{dr}\int_{\Sphere^2}\tilde Hu^{1-\sigma}\, \dif A_h
     -\frac{1}{1-\sigma}\int_{\Sphere^2}\D_{r}\tilde Hu^{1-\sigma}\, \dif A_h
     -\frac{1}{(1-\sigma)}\int_{\Sphere^2}\tilde H^2u^{1-\sigma}\, \dif A_h \\
     =-\int_{\Sphere^2} \left(4\frac{(2-\sigma)}{(3-\sigma)^2} |\nabla u^{\frac{3-\sigma}{2}}|^2
     +B\left(u^{\frac{3-\sigma}{2}}\right)^2\right)\, \dif A_h+\int_{\Sphere^2} Au^{1-\sigma}\, \dif A_h
\end{gather*}
Defining $w=u^{\frac{3-\sigma}{2}}$, we have
\begin{gather*}
     \frac{1}{1-\sigma}\frac{d}{dr}\int_{\Sphere^2}\tilde Hu^{1-\sigma}\, \dif A_h\leq
-\int_{\Sphere^2} \left(4\frac{(2-\sigma)}{(3-\sigma)^2} |\nabla w|^2+B w^2\right)\, \dif A_h\\
 + C \int_{\Sphere^2}\tilde Hu^{1-\sigma}\, \dif A_h.
\end{gather*}
Now, by choosing $\sigma$ close enough to $1$, the middle term can
be made close enough   to $\int_{\Sphere^2} w L w \, \dif A_h$ that it will be negative.
Assuming $\sigma$ to be so chosen, we have
\[
      \frac{d}{dr} \int_{\Sphere^2}\tilde Hu^{1-\sigma}\, \dif A_h\leq C \int_{\Sphere^2} \tilde Hu^{1-\sigma}\, \dif A_h.
\]
Integrating this inequality yields a bound on  $\norm{u}_{L^{1-\sigma}(\Sphere^2)}$ on any finite interval $[r_1,r_2]$.
Taking $\varepsilon=1-\sigma$ in the previous proposition establishes the result.
\end{proof}

\section{Proof of the Main Theorem}
 Collecting the results of sections 1-4 establishes:
\begin{prop}
Given $(\Sphere^2,h_0,\chih,\rho_0)$ as in the hypothesis of the Main Theorem,
there exists $\varepsilon$, functions $u,\rho$ on $[r_0,r_0+\varepsilon]\times\Sphere^2$ with $\rho>0$,
and a family of metrics $h(r), r\in [r_0,r_0+\varepsilon]$ such that the metric
\[
    g=u^2dr^2+h
\]
satisfies $R(g)=16\pi\rho$ and
induces $h_0,\chih$, respectively,  as metric and second fundamental form on $S_{r_0}=\{r_0\}\times\Sphere^2$.
Since $\chih$ is trace free then $S_{r_0}$ is minimal.
\end{prop}
That is to say, we have now constructed the data on the collar region $[r_0,r_0+\varepsilon]$, where it should be noted that the operator  $L=\Delta_h-(\kappa(h)-8\pi\rho)$ is negative at $r_0+\varepsilon$.
We now extend $\rho\geq 0,h$ smoothly to an annular region $[r_0+\varepsilon,T]$ such that
\begin{align*}
    \tilde H&>0,\\
    \rho |_T &\equiv 0,\\
     h|_T &= c\gammab,
\end{align*}
where $c$ is a constant and  $\gammab$ is round.  In addition, we must make the extension in such a way that we may
solve the parabolic scalar curvature equation on this region for any smooth positive initial data at $r_0+\varepsilon$.
By the result of the last section, Theorem~\ref{thm:exist},  we need only ensure that the operator $L$ remains negative.
There are, in general,
 many  extensions that preserve the negativity of $L$.
To make such an extension is not difficult, but is somewhat lengthy and technical,
and we shall not do this here. For a particular example, see the upcoming work~\cite{S3}.
At any rate, with an appropriate extension of $\rho,h,$ and taking initial data for $u$ as given by the metric $g$
already constructed at $r_0+\varepsilon$,
we solve the parabolic scalar curvature equation on $[r_0+\varepsilon,T]$ so that the data is constructed on this region.

At $T$ the local mass density is $0$ and the constructed metric has the form
\[
    g=u^2 dr^2+c\gammab.
\]
One may now follow the constructions presented in previous works, e.g.~\cite{bartnik93},
to construct an asymptotically flat extension of $g$ of  local mass density $0$,
which is such that the foliation of positive mean curvature spheres extends from the sphere at $T$ to $\infty$.   This completes the proof of the main theorem. \hspace{1.25in} $\Box$




\bibliographystyle{amsplain}

\end{document}